\documentclass[12pt,a4paper]{amsart}

\usepackage{amsmath,amsfonts,amssymb}

\usepackage[hmargin=2cm,vmargin=2cm]{geometry}

\usepackage[hyperfootnotes=false,colorlinks=true,linkcolor=blue,%
citecolor=purple,filecolor=magenta,urlcolor=cyan]{hyperref}
\usepackage{nameref,zref-xr}                    

\usepackage{amsthm,amscd}
\usepackage{epsfig}
\usepackage{color}
\usepackage[all]{xy}
\usepackage{tikz}
\usepackage{graphics, setspace}
\usepackage{breqn}
\usepackage{amstext}
\usepackage{array}

\newcommand{\CC}{{\mathbb{C}}}

\newcommand{\QQ}{{\mathbb{Q}}}

\newcommand{\F}[2]{{F_{#1}^{#2}}}

\def\F{{\mathcal F}}

\def\res{{\mathrm{res}}}

\newcommand{\p}{{\partial}}

\newcommand{\ZZ}{\mathbb{Z}}

\newcommand{\bt}{{\bf t}}

\newtheorem{convention}{Convention}[section]
\newtheorem{theorem}{Theorem}[section]
\newtheorem{proposition}[theorem]{Proposition}
\newtheorem{lemma}[theorem]{Lemma}
\newtheorem{corollary}[theorem]{Corollary}

\newtheorem{remark}[theorem]{Remark}

\usepackage{color}

\def\&{\vspace{-5pt}&}

\allowdisplaybreaks

\numberwithin{equation}{section}

\begin{document}

\title{
On B type family of Dubrovin--Frobenius manifolds and their integrable systems
}

\author{Alexey Basalaev}
\address{A. Basalaev:\newline Faculty of Mathematics, HSE University, Usacheva str., 6, 119048 Moscow, Russian Federation, and \newline
Igor Krichever Center for Advanced Studies,
Skolkovo Institute of Science and Technology,
Bolshoy Boulevard 30, bld. 1
Moscow, Russian Federation, 121205}
\email{a.basalaev@skoltech.ru}

\date{\today}

 \begin{abstract}
 According to D.Zuo and an unpublished work of M.Bertola, there is a two--index series of Dubrovin--Frobenius manifold structures associated to a B type Coxeter group. We study the relations between these structures for the different values of these indices. We show that part of the data of such Dubrovin--Frobenius manifold indexed by $(k,l)$  can be recovered by the $(k+r,l+r)$ Dubrovin--Frobenius manifold. Continuing the program of \cite{BDbN} we associate an infinite system of commuting PDEs to these Dubrovin--Frobenius manifolds and show that these PDEs extend the dispersionless BKP hierarchy.
 \end{abstract}
 \maketitle

 
\section{Introduction}

Introduced by B.Dubrovin in the early 90s, Dubrovin--Frobenius manifolds appeared to be important in the different areas of mathematics. The simplest examples of the Dubrovin--Frobenius manifold come from the invariants theory of the simple Coxter group. It was found by B.Dubrovin that the orbit space of every such group has polynomial Dubrovin--Frobenius manifold structure (cf. \cite{D1,D2}).

One of the important applications of the Dubrovin--Frobenius manifolds is in the study of integrable systems. B.Dubrovin and Y.Zhang associated an integrable hierarchy to any Dubrovin--Frobenius manifold given (\cite{DZ}, consider also the other constructions like \cite{B1,B2} and \cite{FGM},\cite{DVV}). 
For the case of simple Coxeter groups the hierarchies of Dubrovin--Zhang were investigated in details in \cite{DLZ}. In particular, the construction of Dubrovin--Zhang for $A_N$ Dubrovin--Frobenius manifold gave the Gelfand--Dikey hierarchy and $D_N$ --- respective Drinfeld--Sokolov hierarchy. For the $B_N$ case one only got the dispersionless Drinfeld--Sokolov hierarchy (see also \cite{LRZ}).

Completely different approach to the construction of an integrable hierarchy with the help of Dubrovin--Frobenius manifolds was introduced in \cite{BDbN}. There the authors assumed an infinite series of Dubrovin--Frobenius manifolds, satisfying some stabilization conditions, in order to construct an infinite system of commuting PDEs. Such stabilization conditions were found for $A$, $B$ and $D$ Dubrovin--Frobenius manifolds, giving dispersionless KP, dispersionless BKP and 1--component reduced 2--component BKP hierarchies respectively. The first important property of this new approach is that the flows are written immediately via the potentials of the Dubrovin--Frobenius manifolds. Another important property is that the compatibility of the PDEs was derived just from the associativity equation of the Dubrovin--Frobenius manifold.
The approach of \cite{BDbN} was later extended in \cite{B22} beyond the theory of Dubrovin--Frobenius manifold.

In this paper we deepen the study of B--type Dubrovin--Frobenius manifolds and their integrable systems.

\subsection{$B_{k,l}$ Dubrovin--Frobenius manifold}
For a $B_l$ Coxeter group and any $k$, s.t. $1 \le k \le l$, D. Zuo\footnote{D.Zuo acknowledges also the unpublished paper of M.Bertola} introduced in \cite{Z07} the structure of a Dubrovin--Frobenius manifold on $M := \CC^{l-1}\times \CC^\ast$. Denote it by $B_{k,l}$ later in the text. The potential of this structure is polynomial in all the variables except one, that comes in both positive and negative powers. Let $\F_{k,l}$ stand for this potential. Renumbering the variables we have
\[
    \F_{k,l} = \F_{k,l}(v^{k+1-l},\dots,v^{-1},v^{0},v^{1},\dots,v^k) \in \QQ[v^{k-l},\dots,v^k]\otimes \QQ[v^{k+1-l},(v^{k+1-l})^{-1}].
\]
The unit of this Dubrovin--Frobenius manifold is given then by $\frac{\p}{\p v^1}$ and the metric $\eta$ satisfies
\begin{equation}\label{eq: eta}
     \eta_{\alpha\beta} = \frac{\p^3 \F_{k,l}}{\p v^1 \p v^\alpha \p v^\beta}
    =
    \begin{cases}
        1/2 \quad & \text{ if } \alpha=k+1-l, \beta=0 \text{ or } \alpha=0, \beta=k+1-l,
        \\
        1/4(l-k) \quad & \text{ if } k-l \le \alpha,\beta \le -1 \text{ and } \alpha+\beta=k+1-l,
        \\
        1/4k \quad & \text{ if } 1 \le \alpha,\beta \le k \text{ and } \alpha+\beta=k+1,
        \\
        0 \quad & \text{otherwise}.
    \end{cases}
\end{equation}

The first theorem of ours is the following stabilization statement. Note that we distinguish between the upper and lower indices of the variables $v^\bullet$ and $v_\bullet$.

\begin{theorem}\label{theorem: stabilization}
    Fix $\alpha,\beta \in \ZZ$. For any $p \ge 1$
    we have the equality in $v_\bullet$
    \[
        \frac{\p^2 \F_{k,l}}{\p v^\alpha \p v^\beta} \mid_{v^\gamma = \eta^{\gamma\delta}v_\delta}
        =
        \frac{\p^2 \F_{k+p,l+p}}{\p v^\alpha \p v^\beta} \mid_{v^\gamma = \eta^{\gamma\delta}v_\delta}.
    \]
    whenever $k \le l$ are such that
    \begin{equation}\label{eq: stabilization conditions}
        \begin{aligned}
        & k \ge \alpha + \beta -1  &&\text{for} \qquad \alpha, \beta \ge 1,
        \\
        & k \ge \alpha, \quad \text{and} \quad l \ge k+1-\beta \qquad &&\text{for} \qquad \alpha \ge 1, \ \beta \le 0,
        \\
        & k \ge 2, \quad \text{and} \quad l \ge k+2-\alpha-\beta \qquad &&\text{for} \qquad \alpha, \beta \le 0.
        \end{aligned}
    \end{equation}
\end{theorem}
\noindent Proof is given in Section~\ref{section: proof 1}.

Note that the change of the variables of the two sides of the equality given in the theorem is different. Namely, the metric $\eta$ should be computed from the different potentials. These metrics do not just have different components, but also have different size.

\subsection{Commuting PDEs}
In \cite{BDbN} and \cite{B22} the authors proved the equalities similar to that of Theorem~\ref{theorem: stabilization} for the $A_l$,$B_l$, $D_l$ Coxeter groups and their ``open extensions''. It was also observed that the associativity equations of the corresponding algebra structures concludes the consistency of the infinite system of PDEs. We extend this result here for $B_{k,l}$ Dubrovin--Frobenius manifolds.

Let $f = f(\dots,t_{-1},t_0,t_1,t_2,\dots)$ be a formal function depending on the variables $t_\bullet$ with both positive and negative indices.
Denote $\p_\alpha := \frac{\p}{\p t_\alpha}$. Fix some $d \ge 1$ and consider the system of PDEs
\begin{equation}\label{eq: PDEs}
    \p_\alpha\p_\beta f = \frac{\p^2 \F_{k_{min},k_{min}+d}}{\p v^\alpha \p v^\beta} \mid_{v^\gamma = \eta^{\gamma\delta} \p_1 \p_\delta f},
    \tag{d-PDEs}
\end{equation}
where $k_{min}$ is minimal index, s.t. $k = k_{min}$ and $l = k_{min}+d$ satisfy Eq.~\eqref{eq: stabilization conditions} for a given pair $\alpha,\beta$.

The right hand side of this equation is a rational function of $\p_1\p_\bullet f$ with the rational coefficients. Theorem~\ref{theorem: stabilization} asserts, that the right hand side is well--defined.
The set $\lbrace \p_1\p_\delta f \rbrace_{\delta=-\infty}^\infty$  should be considered as the initial condition data, and the PDEs above express all the second order derivatives of $f$ via this initial condition data. This is common to consider $\p_1\p_\gamma f$ as the functions of $t_1$ for all $\gamma$. 

\begin{theorem}\label{theorem: consistent}~
    \begin{enumerate}
    \item[(a)] The system \eqref{eq: PDEs} with both positive $\alpha,\beta > 0$ coincides with the dispersionless BKP hierarchy written in Fay form \eqref{bkp-dl}.
    \item[(b)] The system \eqref{eq: PDEs} is consistent.
    \item[(c)] The function $f = \F_{k,k+d} \mid_{v^\gamma = t_\gamma}$ is a solution to \eqref{eq: PDEs} with $\alpha$ and $\beta$ satisfying Eq.~\eqref{eq: stabilization conditions} for the given $k$ and $l=k+d$.
    \item[(d)] The system \eqref{eq: PDEs} with both negative $\alpha,\beta < 0$ is equivalent to the dispersionless BKP hierarchy.
    \end{enumerate}
\end{theorem}
\noindent Proof is given in Section~\ref{section: proof 2}.

\subsection{Dispersionless BKP hierarchy}
BKP hierarchy was introduced in \cite{DJKM} via the Lax form. It can be rewritten in an equivalent form via the Fay--type identities on the tau--function $\tau  = \tau(\bt)$ (see \cite{T1,T2}). Assume the logarithm of the tau--function to have the $\hbar$--expansion ${\log \tau = \sum_{g \ge 0} F_g \hbar^{g-2}}$. Then the dispersionless limit of the BKP hierarchy is the following system of PDEs on $F_0$.
\begin{align}
    \left( 1 - \frac{\p_1 (2D^{\mathrm{B}}(z_1) + 2D^{\mathrm{B}}(z_2))F_0}{z_1+z_2} \right) e^{2D^{\mathrm{B}}(z_1)\cdot2D^{\mathrm{B}}(z_2) F_0}
     = 1 - \frac{\p_1 (2D^{\mathrm{B}}(z_1) - 2D^{\mathrm{B}}(z_2))F_0}{z_1-z_2},
\tag{BKP-dl}\label{bkp-dl}
\end{align}
where
\[
    D^{\mathrm{B}}(z) := \sum_{n \ge 0} \frac{z^{-2n-1}}{2n+1} \frac{\p}{\p t_{n+1}}.
\]
This equality should be understood in the ring of formal power series in $z_1$ and $z_2$.
\begin{remark}
Original BKP hierarchy is written in the odd--index coodinates $\widetilde t_\bullet$ only. This is essential due to the connection of BKP and KP hierarchies. However, this does not play any role in our text and we choose our coordinates so that $t_k := \widetilde t_{2k-1}.$
\end{remark}
%

\subsection{Examples}
Denote $f_{\alpha,\beta} := \p_\alpha\p_\beta f$. The flows of \eqref{eq: PDEs} with  $\beta=1$ are just $f_{1,\alpha} = f_{1,\alpha}$. The more complicated flows read.
\begin{equation}\label{eq: first positive flows}
 \begin{aligned}
    f_{2,2} &= \frac{4}{3} f_{1,1}^3 - 2 f_{1,2} f_{1,1}+f_{1,3},
    \\
      f_{2,3} &= 4 f_{1,2} f_{1,1}^2-2 f_{1,3} f_{1,1}-2 f_{1,2}^2+f_{1,4},
    \\
      f_{2,4} &= 4 f_{1,3} f_{1,1}^2+4 f_{1,2}^2 f_{1,1}-2 f_{1,4} f_{1,1}-4 f_{1,2} f_{1,3}+f_{1,5}.
\end{aligned}
\end{equation}
These are dispersionless BKP flows. We also have for $d=1,2$ the additional flows
\begin{align*}
    f_{0,2} &= 2 f_{1,0} f_{1,1},
    \\
    f_{0,3} &= 2 f_{1,0} ( 2f_{1,1}^2 + f_{1,2}),
    \\
    f_{0,4} &=  2 f_{1,0} ( 4f_{1,1}^3 + 4 f_{1,2} f_{1,1} + f_{1,3}),
\end{align*}
\begin{align*}
    f_{-1,2} &= \frac{16}{3}f_{1,0}^3 + 2 f_{1,-1} f_{1,1},
    \\
    f_{-1,3} &= 32 f_{1,1} f_{1,0}^3+4 f_{1,-1} f_{1,1}^2 + 2 f_{1,-1} f_{1,2}, \qquad f_{0,0} = \frac{f_{1,-1}}{2 f_{1,0}}.
\end{align*}
plus infinitely more flows on $f_{-1,\alpha}$ and $f_{0,\alpha}$ with $\alpha \ge 1$.

%

\subsection{Comments and conjectures}
In order to be more explicit, to prove (d) of Theorem~\ref{theorem: consistent} we show that the system \eqref{eq: PDEs} with both negative $\alpha,\beta < 0$ coincides with the Fay form of dispersionless BKP hierarchy after the certain polynomial change of the initial condition data $\lbrace \p_1\p_\delta f \rbrace_{\delta=-\infty}^0$.

We do not identify the flows with $\alpha \cdot \beta < 0$, however the theorem above makes it essential to conjecture that the full system Eq.~\eqref{eq: PDEs} for $B_{k,l}$ Dubrovin--Frobenius manifolds is equivalent to the 2--component dispersionless BKP hierarchy.

We also conjecture that the system Eq.~\eqref{eq: PDEs} is equivalent to the full 2--component BKP hieararchy after the certain $\hbar$--deformation of the differential operators $\p_k \to \p_k^\hbar = \p_k + O(\hbar)$ like it was done in \cite{NZ} for KP hieararchy. However, this requires to compute the higher genera potentials for $B_{k,l}$ Dubrovin--Frobenius manifold what is a complicated task.

\subsection{Connection to DR and DZ hierarchies}
Our approach in studying PDEs with the help of Dubrovin--Frobenius manifolds is different from DR and DZ by construction. In particular, we need an infinite family of compatible potentials, wherear DR and DZ hieararchies are associated to one Dubrovin--Frobenius manifold or one CohFT respectively. In particular Eq.\eqref{eq: PDEs} for A--type Dubrovin--Frobenius manifolds give a dispersionless KP hieararchy while DR and DZ give the reduction of KP hieararchy. Note that dispersionless limits of DR and DZ hieararchies are equivalent. However one can make the following comment.

DZ hieararchy is constructed via the certain quasi--Miura transformation from the so--called principal hieararchy. Given an $N$--dimensional Dubrovin--Frobenius manifold with the potential $\F = \F(v^1,\dots,v^N)$, this is the system of differential equations on the set of functions $(w^1,\dots,w^N)$. Each $w^\alpha$ is the function of $t^{\beta,p}$ with $1 \le \beta \le N$ and $p \in \ZZ_{\ge 0}$. The dependence of $w^\alpha$ on $t^{\beta,0}$ is given by
\[
    \frac{\p w^\alpha}{\p t^{\beta,0}} = \eta^{\alpha \gamma} \frac{\p^3 \F}{\p v^{\gamma} \p v^{\beta} \p v^{\delta}} \mid_{v^\mu = \omega^\mu} \frac{\p w^{\delta}}{\p t^{1,0}}.
\]
It's easy to see that taking $\omega^\alpha := \eta^{\alpha \gamma} \p_1\p_\gamma f$ transforms this equation to the $\frac{\p }{\p t^{1}}$--derivative of the first $N$ flows of Eq.~\eqref{eq: PDEs}. However, this is not enough to claim that dispersionless limits of DZ and DR hierarchies are reductions of 2--component BKP.

\subsection{Acknowledgements}
The work of Alexey Basalaev was supported by the Russian Science Foundation (grant No. 24-11-00366).

The author acknowledges no conflict of interest or data used for this research.

\section{$B_{k,l}$ Dubrovin--Frobenius manifold}

This section is devoted to the $B_{k,l}$ Dubrovin--Frobenius manifold structure. The first two subsection recall the definition and construction of this Dubrovin--Frobenius manifold following \cite{D2} and \cite{Z07} respectively. The other subsections are new, providing some finer information about $B_{k,l}$.

Working with the $B_{k,l}$ Dubrovin--Frobenius manifold we always keep the notation consistent with that of \cite{Z07}. Due to this the coordinates $t^\bullet$ in this section are totally different from the coordinates $t_\bullet$ used in Introduction and Eq.~\eqref{eq: PDEs} in particular.

\subsection{Dubrovin--Frobenius manifolds}
Assume $M$ to be an open full--dimensional subspace of $\CC^l$. We say that it's endowed with a structure of Dubrovin--Frobenius manifold if there is a regular function $\F = \F(t_1,\dots,t_l)$ on $M$, s.t. the following conditions hold (cf. \cite{D2}).

\begin{itemize}
	\item 
There is a distinguished variable $t_k$, for some $1 \le k \le L$, such that:
\[
    \frac{\p \F}{\p t_k} = \frac{1}{2} \sum_{\alpha,\beta = 1}^l \eta_{\alpha,\beta} t_\alpha t_\beta,
\]
and $\eta_{\alpha,\beta}$ are components of a non-degenerate bilinear form $\eta$ (which does not depend on $t_\bullet$). In what follows denote by $\eta^{\alpha,\beta}$ the components of $\eta^{-1}$.

\item The function $\F$ satisfies a large system of PDEs called the WDVV equations:
\[
\sum_{\mu,\nu = 1}^l \frac{\p^3 \F}{\p t_\alpha \p t_\beta \p t_\mu} \eta^{\mu,\nu} \frac{\p^3 \F}{\p t_\nu \p t_\gamma \p t_\sigma}
=
\sum_{\mu,\nu = 1}^l \frac{\p^3 \F}{\p t_\alpha \p t_\gamma \p t_\mu} \eta^{\mu,\nu} \frac{\p^3 \F}{\p t_\nu \p t_\beta \p t_\sigma},
\]
which should hold for every given $1\leq\alpha,\beta,\gamma,\sigma\leq l$.

\item There is a vector field $E$ called the \textit{Euler vector field}, s.t. modulo quadratic terms in $t_\bullet$ we have $E \cdot \F = (3-\delta) \F$ for some fixed complex number $\delta$. We will assume $E$ to have the following simple form
\[
    E = \sum_{i=1}^l d_i t_i \frac{\p}{\p t_i}
\]
for some fixed numbers $d_1,\dots,d_l$. 
\end{itemize}

Given such a data $(M,\F,E)$ one can endow every tangent space $T_pM$ with a structure of commutative associative product $\circ$ (depending on $\bt$) defined as follows:
\[
    \frac{\p}{\p t_\alpha} \circ \frac{\p}{\p t_\beta} = \sum_{\delta,\gamma=1}^l\frac{\p^3 \F}{\p t_\alpha \p t_\beta \p t_\delta} \eta^{\delta\gamma} \frac{\p}{\p t_\delta}.
\]
The unit of this product is the vector field $e = \frac{\p}{\p t_k}$.
It follows that $\eta(a \circ b,c) = \eta(a,b \circ c)$ for any vector fields $a,b,c$.

B.Dubrovin proposed several \textit{sources} of Dubrovin--Frobenius manifolds. In particular, we gave the construction of a Dubrovin--Frobenius manifold structure on the complexified orbit space of a simple Coxeter group (see Lecture 4 of \cite{D1}). For the $B$--type Coxeter group this construction can be generalized as follows.

\subsection{Construction of $B_{k,l}$ Dubrovin--Frobenius manifold}
Fix $k,l$, s.t. $l \ge 1$ and $1 \le k \le l$. Let $x^1,\dots,x^l$ be the orthonormal coodinates of $\CC^l$ equipped with the scalar product $( \cdot ,\cdot  )$.

Let $\sigma_i$ be the $i$--th elementary symmetric function and $y^i := \sigma_i((x^1)^2,\dots,(x^l)^2)$. Then $y^1,\dots,y^l$ are coordinates on the orbit space of $B_l$.

Denote by $g$ the matrix with the following components
\[
    g^{ij} := (dy^i,dy^j) = \sum_{p=1}^l \frac{\p y^i}{\p x^p}\frac{\p y^j}{\p x^p}.
\]
It defines the metric on the orbit space of $B_l$.
For $P_l := u^{2l} + \sum_{j=1}^l u^{2(l-j)}y^j$ we have (cf.~Proposition 2.2.2 of \cite{SYS})
\begin{align}
    \sum_{i,j=1} g^{ij} u^{2(l-i)}v^{2(l-j)} &= \frac{2}{u^2-v^2} (u P_l'(u)P_l(v) - vP_l(u)P_l'(v)),
    \label{eq: g via generating function}
\end{align}

The orbit space of $B_l$ carries another bilinear form $\eta$ with the components
\[
    \eta^{ij} := \frac{\p g^{ij}}{\p y^k},
\]
whose determinant is a constant multiple of $(y^l)^{l-k}$. In particular, $\eta$ is nondegenerate on $M := \lbrace (y^1,\dots,y^l) \in \CC^l \ | \ y^l \neq 0 \rbrace$.
It follows immediately that $\eta$ is block--diagonal with the $k \times k$ and $(l-k)\times (l-k)$ blocks.

\begin{convention}\label{convention on indices}
In what follows $g^{\bullet,\bullet}$ will stay for the components of the metric $g$ in $dt^\bullet$ basis if the indices are Greek and in the $dy^\bullet$ basis if the indices are Latin.
\end{convention}

The following theorem summarizes Sections 2 and 3 of \cite{Z07}.

\begin{theorem}[\cite{Z07}]\label{theorem: Bkl construction}
    $M$ carries the flat coordinates $t^1,\dots,t^l$, s.t. $\eta$ is constant in these coordinates with the components given by Eq.~\eqref{eq: eta}. The components of $g$ in $dt^\bullet$ basis can be integrated by the function $\F_{k,l}=\F_{k,l}(t^1,\dots,t^l)$:
    \[
        \frac{g^{\alpha\beta}}{\widetilde d_\alpha + \widetilde d_\beta} = \eta^{\alpha\gamma}\eta^{\beta\delta}\frac{\p^2 \F_{k,l}}{\p t^\gamma \p t^\delta},
    \]
    with
    \[
        \widetilde d_\alpha = \frac{2\alpha -1}{2k}, \
        \widetilde d_\gamma = \frac{2(l+1-\gamma) -1}{2(l-k)},
        \quad \alpha \le k, \ \gamma > k
    \]
    The function $\F_{k,l}$ depends polynomially on $t^1,\dots,t^{l-1}$ and is a Laurent polynomial in $t^l$. The variable $t^k$ is distinguished by
    \[
        \eta_{\alpha\beta} = \frac{\p^3 \F_{k,l}}{\p t^k \p t^\alpha \p t^\beta}.
    \]
    Moreover $\F_{k,l}$ is the solution to WDVV equation with $\eta$ above. The following quasihomogeneity condition holds
    \begin{equation}\label{eq: potential quasihomogeneity}
        E \cdot \F_{k,l} = \frac{2k+1}{k} \F_{k,l} \quad \text{for} \quad E := \sum_{\alpha=1}^l \left(\widetilde d_\alpha + \frac{1}{2k} \right) t^\alpha \frac{\p }{\p t^\alpha}.
    \end{equation}
    Summing up we have that $(M,\F_{k,l},E)$ is a Dubrovin--Frobenius manifold with the unit $e = \frac{\p}{\p t^k}$.
\end{theorem}

\begin{remark}
    For Theorem~\ref{theorem: stabilization} we use the same potential $\F_{k,l}$, but written in the other coordinates $v^\alpha$ given by $v^{\alpha} = t^{k+1-\alpha}$. In this coordinates $\F_{k,l} = \F_{k,l}(v^{k+1-l},\dots,v^{-1},v^0,v^1,\dots,v^k)$. In the language of Dubrovin--Frobenius manifolds we have that $t^k = v^1$ is the coordinate, corresponding to the unit vector field.
    
    We will come to the coordinates $v^\alpha$ only when proving Theorem~\ref{theorem: stabilization}. Even though these change of coordinates is just an index shift, it is crucial for the statement of the theorem.
\end{remark}

The construction of B.Dubrovin (cf. \cite{D1,D2}) for $B_l$ Coxeter group is obtained by taking $k = l$. In this case the determinant of $\eta$ is a non--zero constant and $M$ can be taken to be $\CC^l$.

It follows immediately from the block--diagonal form of $\eta$ and also from the quasihomogeneity of the whole construction that the flat coordinates $t^\bullet$ above, are expressed via the coordinates $y^\bullet$ so that
\begin{equation}\label{eq: flat coordinate via y}
\begin{aligned}
    t^\alpha = t^\alpha(y^1,\dots,y^\alpha) &\quad \text{ if } \quad \alpha \le k,
    \\
    t^\beta = t^\gamma(y^\gamma,\dots,y^l) &\quad \text{ if } \quad \gamma \ge k+1.
\end{aligned}
\end{equation}

Moreover the coordinates $t^1,\dots,t^k$ are expressed via $y^1,\dots,y^k$ by the same formulae as in Dubrovin's $B_k$ construction\footnote{note the group rank to be $k$ and not $l$}.

\subsection{Examples}
Let the subscript $(l)$ indicate the rank of the group to which $g$ corresponds.
In the $dy$ basis we have
\[
    g_{(4)}
    =
    \left(
\begin{array}{cccc}
 4 y^1 & 8 y^2 & 12 y^3 & 16 y^4 \\
 8 y^2 & 4 y^1 y^2+12 y^3 & 8 y^1 y^3+16 y^4 & 12 y^1 y^4 \\
 12 y^3 & 8 y^1 y^3+16 y^4 & 4 y^2 y^3+12 y^1 y^4 & 8 y^2 y^4 \\
 16 y^4 & 12 y^1 y^4 & 8 y^2 y^4 & 4 y^3 y^4 \\
\end{array}
\right),
\]
The corresponding flat coordinates are given by
\begin{align*}
 & y^1 = t^1, \ y^2 = t^2 + \frac{(t^1)^2}{4}, \ y^3 = t^3 + \frac{(t^1)^3}{108} + \frac{t^1 t^2}{6}, \quad y^4 = (t^4)^2 \qquad && k=3,
 \\
 & y^1 = t^1, \ y^2 = t^2 + \frac{(t^1)^2}{8}, \quad y^3 = t^3 t^4, \ y^4 = (t^4)^4, \qquad && k=2,
 \\
 & y^1 = t^1, \quad y^2 = t^2 t^4 + \frac{1}{12} (t^3)^2,  \ y^3 = t^3 (t^4)^3, \ y^4 = (t^4)^6, \qquad && k=1.
\end{align*}

\[
    g_{(5)} = \left(
\begin{array}{ccccc}
 4 y^1 & 8 y^2 & 12 y^3 & 16 y^4 & 20 y^5 \\
 8 y^2 & 4 y^1 y^2+12 y^3 & 8 y^1 y^3+16 y^4 & 12 y^1 y^4+20 y^5 & 16 y^1 y^5 \\
 12 y^3 & 8 y^1 y^3+16 y^4 & 4 y^2 y^3+12 y^1 y^4+20 y^5 & 8 y^2 y^4+16 y^1 y^5 & 12 y^2 y^5 \\
 16 y^4 & 12 y^1 y^4+20 y^5 & 8 y^2 y^4+16 y^1 y^5 & 4 y^3 y^4+12 y^2 y^5 & 8 y^3 y^5 \\
 20 y^5 & 16 y^1 y^5 & 12 y^2 y^5 & 8 y^3 y^5 & 4 y^4 y^5 \\
\end{array}
\right).
\]

The corresponding flat coordinates are given by
\begin{align*}
    &y^1 = t^1, \ y^2 = t^2 + \frac{5 (t^1)^2}{16}, \ y^3 = t^3 + \frac{(t^1)^3}{32} + \frac{3 t^1 t^2}{8},
    \\
    & \qquad\qquad y^4 = t^4 + \frac{(t^1)^4}{2048} + \frac{(t^1)^2t^2}{64}  +\frac{ t^1t^3}{8}+\frac{(t^2)^2}{16}, \quad y^5 = (t^5)^2, \qquad && k=4,
    \\
    &y^1= t^1, \ y^2= t^2 + \frac{(t^1)^2}{4}, \ y^3= t^3 + \frac{(t^1)^3}{108} + \frac{t^2 t^1}{6}, \qquad y^4 = t^4 t^5, \ y^5 = (t^5)^4 \qquad && k=3,
    \\
    &y^1= t^1, \ y^2= t^2 + \frac{(t^1)^2}{8}, \qquad y^3= t^3 t^5 + \frac{(t^4)^2}{12}, \ y^4 = t^4 (t^5)^3, \ y^5 = (t^5)^6, \qquad && k=2,
    \\
    &y^1 = t^1 \qquad y^2= t^2t^5 + \frac{t^3 t^4}{8}, \ y^3 = t^3 (t^5)^3 + \frac{3}{16} (t^4)^2 (t^5)^2, \ y^4= t^4 (t^5)^5, \ y^5= (t^5)^8, \qquad && k=1.
\end{align*}

The examples of the potentials $\F_{k,l}$ are\footnote{we follow Dubrovin's convention using only lower indices in the examples}.

\begin{align*}
 & \F_{2,5} = \frac{t_5^6}{10}+\frac{t_1 t_4 t_5^3}{6} +\frac{t_1^2 t_3 t_5}{16} +\frac{t_2 t_3 t_5}{2} +\frac{t_1^5}{7680}+\frac{t_1 t_2^2}{16} +\frac{t_1^2 t_4^2}{192} +\frac{t_2 t_4^2}{24} +\frac{t_3^2 t_4}{24 t_5}-\frac{t_3 t_4^3}{216 t_5^2}+\frac{t_4^5}{4320 t_5^3},
\\
& \F_{3,6} = \frac{t_1^7}{3265920}+\frac{t_2^2 t_1^3}{2592}+\frac{t_5^2 t_1^3}{2592}+\frac{t_4 t_6 t_1^3}{216} +\frac{t_5 t_6^3 t_1^2}{24} + \frac{t_6^6 t_1}{10} -\frac{t_2^3 t_1}{432} +\frac{t_3^2 t_1}{24}
\\
&\quad +\frac{t_2 t_5^2 t_1}{144} +\frac{t_2 t_4 t_6 t_1}{12} +\frac{t_2 t_5 t_6^3}{6} +\frac{t_3 t_5^2}{24} +\frac{t_2^2 t_3}{24} +\frac{t_3 t_4 t_6}{2} +\frac{t_4^2 t_5}{24 t_6}-\frac{t_4 t_5^3}{216 t_6^2}+\frac{t_5^5}{4320 t_6^3}.
\end{align*}

\subsection{Some zeros}
The following proposition gives some insight on the structure of $B_{k,l}$ Dubrovin--Frobenius manifold.

\begin{proposition}\label{prop: 3rd derivatives zeros}
Let $k \ge 2$.
For any $\gamma,\delta,\alpha,\beta$, s.t. $k+1 \le \gamma,\delta \le l$ and $1 \le \alpha,\beta \le k$ we have
\begin{enumerate}
 \item[(a)] if $\gamma+\delta \le k+l$, then
 \[
    \frac{\p^3 \F_{k,l}}{\p t^\alpha \p t^\gamma \p t^\delta} = 0,
 \]
 \item[(b)] if $\alpha+\beta \ge k+1$, then
 \[
 \frac{\p^3 \F_{k,l}}{\p t^\gamma \p t^\alpha \p t^\beta} = 0.
 \]
\end{enumerate}
\end{proposition}
\begin{proof}
    For case (a) we have to prove the vanishing of
    \[
        \frac{\p}{\p t^\alpha} g^{k+l+1-\gamma,k+l+1-\delta} = \sum \frac{\p y^c}{\p t^\alpha} \frac{\p}{\p y^c} \left( \frac{\p t^{k+l+1-\gamma}}{\p y^a} \frac{\p t^{k+l+1-\delta}}{\p y^b} g^{ab} \right),
    \]
    where the summation is taken over $1 \le c \le \alpha$, $l \ge  a \ge k+l+1-\gamma$, $l \ge b \ge k+l+1-\delta$. By using Eq.~\eqref{eq: flat coordinate via y} we get
    \[
        \frac{\p}{\p t^\alpha} g^{k+l+1-\gamma,k+l+1-\delta} = \sum \frac{\p y^c}{\p t^\alpha} \frac{\p t^{k+l+1-\gamma}}{\p y^a} \frac{\p t^{k+l+1-\delta}}{\p y^b} \frac{\p g^{ab}}{\p y^c} .
    \]
    To prove the proposition we show that $\p g^{ab}/\p y^c = 0$ under the conditions on the indices $a,b,c$ above.

    It follows from Eq.~\eqref{eq: g via generating function} that $g^{ab}$ is at most quadratic in $y^\bullet$. Moreover if one assigns the grading to $y^p$ by $\deg y^p = p$ then $\deg g^{ab} = a+b-1$. We have
    \[
        \deg \frac{\p g^{ab}}{\p y^c} \ge 2(k+l+1) - \gamma - \delta -1 - c \ge k+l+1-c.
    \]
    If $g^{ab}$ is quadratic and the partial derivative above is non--zero, than its degree should be the degree of some $y^\bullet$. With the bounds on $c$ above this can not be ever reached.

    The proof of claim (b) is completely similar and therefore skipped.
\end{proof}
\begin{remark}
    One notes immediately by looking at $\F_{2,5}$ and $\F_{3,6}$ above that the conditions on $\alpha+\beta$ in (a) and $\gamma+\delta$ in (b) are strict.
\end{remark}

\begin{corollary}\label{corollary: Bk in Bkl}
    For $\alpha+\beta \le k+1$ we have
    \[
        \frac{\p^2 \F_{k,l}}{\p v^\alpha \p v^\beta} = \frac{\p^2 \F_{k,k}}{\p v^\alpha \p v^\beta}.
    \]
\end{corollary}
\begin{proof}
    This follows immediately that for $1 \le \alpha,\beta \le k$ in the flat basis the components $g^{\alpha,\beta}$ for $B_{k,l}$ coincide with the components $g^{\alpha,\beta}$ for $B_k$ up to a function depending on $t^\gamma$ with $\gamma \ge k+1$. The proposition above shows that this dependence is trivial if $\alpha+\beta \le k+1$. The statement follows now from the relation between $g$ and $\F_{k,l}$.
\end{proof}

\subsection{More on the metric $g$}
In what follows denote $\widehat y^a := y^ly^{l-a}$ for any $a =1, \dots, l-1$ and $\widehat y^l = y^l$.

\begin{proposition}\label{prop: g neg-pos}
    In the basis $dy^1,\dots,dy^l$ holds
    \[
        g^{a,b}(y) = (y^l)^{-2} g^{l+1-a,l+1-b}(\widehat y),
    \]
    for any $2 \le a,b \le l-1$.
\end{proposition}
\begin{proof}
    We use Eq.~\eqref{eq: g via generating function}. Direct computations show that
    \begin{align*}
        & y^l P_l(u,y) = u^{2l} P_l(u^{-1},\widehat y) - 1 + (y^l)^2,
        \\
        & y^l P'_l(u,y) = 2l \cdot u^{2l-1} P(u^{-1},\widehat y) - u^{2l-2} P'(u^{-1},\widehat y),
    \end{align*}
    where $P_l'(u^{-1},\widehat y)$ stands for the change of the variables $u \to u^{-1}$ applied to the polynomial $\frac{\p P_l}{\p u}$. Let $W_l = W_l(u,v)$ stand for the left hand side generating function of Eq.~\eqref{eq: g via generating function}. We have the equality of power series in $u,v$
    \begin{align*}
        &\left(W(u,v,y)- (y^l)^{-2} u^{2l-2}v^{2l-2} W(u^{-1},v^{-1},\widehat y)\right) (u^2-v^2)
        \\
        &\quad = ((y^l)^{-2}-1)\left(2l u^{2l}P_l(u^{-1},\widehat y) - 2l v^{2l}P_l(v^{-1},\widehat y) - u^{2l-1}P'_l(u^{-1},\widehat y)+v^{2l-1}P'_l(v^{-1},\widehat y) \right).
    \end{align*}
    The right hand side power series only involves the terms like $u^{2p}$ and $v^{2q}$, hence the left hand side power does not have any mixed monomials $u^{2p}v^{2q}$. Comparing the respective coefficients we get the desired equality.

\end{proof}

\subsection{More on flat coordinates}

It was observed in \cite[Section~5]{Z07} that the flat coordinates $t^\alpha$ of $B_{k,l}$ can be found as follows. Consider an auxiliary function $\lambda: \CC \times M \to \CC$
\[
    \lambda := p^{2k}+\sum_{a=1}^l p^{2(k-a)} y^a.
\]
For every $y \in M$, assumed as a function in $p$, the function $\lambda$ has order $2k$ and $2(l-k)$ poles at $\infty$ and $0$ respectively. Then the flat coordinates of the metric $\eta$ read\footnote{we correct here a missprint Eq.(5.7) in \cite{Z07} where the variables are listed in a wrong order}
\begin{align}
    &t^\alpha = \frac{2k}{2\alpha -1} \res_{p=\infty} \left( \lambda^{\frac{2\alpha-1}{2k}} dp\right), \quad 1 \le \alpha \le k, \label{eq: flat positive}
    \\
    &t^{l+1-\beta} = \frac{2(l-k)}{2\beta -1} \res_{p=0} \left( \lambda^{\frac{2\beta-1}{2(l-k)}} dp\right), \quad 1 \le \beta \le l-k, \label{eq: flat negative}    
    \\
    &t^{l} = \res_{p=0} \left( \lambda^{\frac{1}{2(l-k)}} dp\right) = \left(y^l\right)^{1/2(l-k)}, \label{eq: y_l flat}
\end{align}
where $\lambda^{\frac{1}{2k}}$ and $\lambda^{\frac{1}{2(l-k)}}$ should be assumed as the formal power series in $p$ with coefficients, depending in $y$. 

\begin{proposition}\label{prop: psi neg-pos}
    Let $l=2k$ and $\alpha \le k$. Consider $t^\alpha$ as the functions in $y^1,\dots,y^l$. Then we have the following equality
    \[
        \frac{\p t^\alpha}{\p y^b}(y) = -r_\alpha (t^l)^{2(k-\alpha)+1} \frac{\p t^{l+1-\alpha}}{\p y^{l+1-b}}(\widehat y),
    \]
    where $r_1 = l$ and $r_\alpha = 1$ if $\alpha > 1$.
\end{proposition}
\begin{proof}
    If $\alpha = 1$ the proposition follows by Eq.~\eqref{eq: y_l flat}. Assume $\alpha > 1$.
    Note that 
    \[
        (y^l)^{-1} \lambda(p^{-1},\widehat y) = p^{2(l-k)} + \sum_{b=1}^{l-1} y^b p^{2(l-k-b)} + p^{-2k}(y^l)^{-1}.
    \]
    For $l=2k$ in a neighborhood of $p=\infty$ by Eq.\eqref{eq: flat positive} we have
    \begin{align*}
        \frac{\p t^\alpha}{\p y^b} &= \frac{1}{2k} \res_{p =\infty} \left(p^{2(k-b)} \left( \lambda(p,y) \right)^{\frac{2(\alpha-k)-1}{2k}} dp \right)
        \\
        &\quad = \frac{1}{2k} \res_{p =\infty} \left(p^{2(k-b)} \left( (y^l)^{-1} \lambda(p^{-1},\widehat y) \right)^{\frac{2(\alpha-k)-1}{2k}} dp \right)
        \\
        &\quad\quad = - \frac{1}{2k} (y^l)^{\frac{2(k-\alpha)+1}{2k}} \res_{p =0} \left(p^{2(b-k)-2} \left( \lambda(p,\widehat y) \right)^{\frac{2(\alpha-k)-1}{2k}} dp \right),
    \end{align*}
    what concludes the proof by using Eq.~\eqref{eq: y_l flat} and \eqref{eq: flat negative}.
\end{proof}

Let $\widehat t = \widehat t(t)$ be the change of the variables induced by the change of the variables $\widehat y(y)$.

\begin{proposition}\label{proposition: positive to negative}
    For $l=2k$ and any $1 \le \alpha,\beta \le k$ we have
    \[
        \frac{\p^2 \F_{k,2k}}{\p t^{k+1-\alpha} \p t^{k+1-\beta}}(t)
        =
        \left( t^l \right)^{-2(\alpha+\beta-1)} \frac{\p^2 \F_{k,2k}}{\p t^{k+\alpha} \p t^{k+\beta}}(\widehat t),
    \]
    where $r_\alpha = 4k$ if $\alpha = 1$ and $r_\alpha = 1$ otherwise.
\end{proposition}
\begin{proof}
    We make use of Theorem~\ref{theorem: Bkl construction}. First of all note that for $l=2k$ we have $\widetilde d_\gamma = \widetilde d_{l+1-\gamma}$ for any index $\gamma$. Then we have
    \[
     \frac{\p^2 \F_{k,2k}}{\p t^{k+1-\alpha}\p t^{k+1-\beta}} = \frac{1 }{(4k)^2} \frac{g^{\alpha,\beta}}{\widetilde d_{\alpha}+\widetilde d_{\beta}}
     = 
     \frac{(4k)^{-2}}{\widetilde d_{\alpha}+\widetilde d_{\beta}} \sum_{a,b} \frac{\p t^{\alpha}}{\p y^a}\frac{\p t^{\beta}}{\p y^b} g^{a,b}
    \]
    Applying Propositions~\ref{prop: g neg-pos} and \ref{prop: psi neg-pos} expression above reads
    \[
     \frac{(4k)^{-2} r_{\alpha} r_{\beta}}{\widetilde d_{l+1-\alpha}+\widetilde d_{l+1-\beta}} \left( t^l \right)^{-2(\alpha+\beta-1)} \left( \sum_{a,b} \frac{\p t^{l+1-\alpha}}{\p y^{l+1-a}} \frac{\p t^{l+1-\beta}}{\p y^{l+1-b}} g^{l+1-a,l+1-b} \right)(\widehat t).
    \]
    Note that $\eta_{\alpha,\gamma} = (4k)^{-1} \delta_{\alpha+\gamma,k+1} = (4k)^{-1} \delta_{2k+\alpha+\gamma,l+k+1} = r_\alpha r_\gamma \eta_{l+1-\alpha,l+1-\gamma}$.
    This concludes the proof.
\end{proof}

\begin{proposition}\label{proposition: positive rescaling}
    Let $l=2k$ and $\alpha+\beta \le k+1$. Consider the change of the variables $t^\gamma = (t^l)^{2\gamma} \overline t^\gamma$ for $\gamma < l$ and $t^l = \overline t^l$. Then
    \[
        \left( (t^l)^{-2(\alpha+\beta-1)}\frac{\p^2 \F_{k,2k}}{\p t^{k+1-\alpha}\p t^{k+1-\beta}} \right)(t(\overline t)) 
        = 
        \frac{\p^2 \F_{k,2k} (\overline t)}{\p \overline t^{k+1-\alpha}\p \overline t^{k+1-\beta}}.
    \]
    Namely, rescaling the variables $t^1,\dots,t^{l-1}$ as above in the second order derivatives $\frac{\p^2 \F_{k,2k}}{\p t^{k+1-\alpha}\p t^{k+1-\beta}}$ results just in the multiple $(t^l)^{2(\alpha+\beta-1)}$.
\end{proposition}
\begin{proof}
    It following from Proposition~\ref{prop: 3rd derivatives zeros} that the second order derivative in question depends on $t^1,\dots,t^k$ only. According to the quasihomogeneity condition (see Eq.~\eqref{eq: potential quasihomogeneity}) we have
    \[
        E \cdot \frac{\p^2 \F_{k,2k}}{\p t^{k+1-\alpha}\p t^{k+1-\beta}} = \frac{2k+1 - 2(k+1) + \alpha+\beta}{k} \frac{\p^2 \F_{k,2k}}{\p t^{k+1-\alpha}\p t^{k+1-\beta}}.
    \]
    Then the monomial $\prod_{i \in I} t^{\gamma_i}$ is a summand of this second order derivative only if $\sum_{i \in I} \gamma_i = \alpha+\beta -1$. This concludes the proof.
\end{proof}

\section{Proofs of the theorems}

\subsection{Proof of Theorem~\ref{theorem: stabilization}}\label{section: proof 1}
Let $\alpha,\beta \in \ZZ$. Divide the proof into three parts. 1st: $\alpha,\beta \ge 1$, 2nd: $\alpha,\beta \le 0$ and 3rd: $\alpha \le 0$ and $\beta \ge 0$. We call these parts PP, NN and PN sectors respectively, keeping P for positive and N for negative.

In what follows we rewrite the stabilization statement of the theorem in the coordinates $t^\bullet$. In particular, we will have no negative indices.

\subsubsection{Sector PP}
This case follows from  Corollary~\ref{corollary: Bk in Bkl} above and Proposition~4.3 of \cite{BDbN}.

\subsubsection{Sector NN}\label{section: NN proof}
Let $1 \le \gamma,\delta \le l-k$ s.t. $\gamma+\delta \le l-k$.
Stabilization in NN sector is equivalent to
\begin{equation*}
    \frac{\p^2 \F_{k,l}}{\p t^{k+\gamma} \p t^{k+\delta}} (\widetilde t)
    =
    \frac{\p^2 \F_{k+r,l+r}}{\p t^{k+r+\gamma} \p t^{k+r+\delta}} (\widetilde t)
\end{equation*}
where $\widetilde t = \widetilde t(t)$ is the change of the variables given by $t^\epsilon = 4k \cdot \widetilde t^\epsilon$ if $1 \le \epsilon \le k$ and $t^l = 2 \cdot \widetilde t^{k+1-l}$, $t^{k+1} = 2 \cdot \widetilde t^0$,  $t^\nu = 4(l-k) \cdot \widetilde t^{k+1-\nu}$ for $\nu \ge k+1$.

The equation above is equivalent to
\[
    \frac{g_{(l)}^{l+1-\gamma,l+1-\delta}(\widetilde t)}{\widetilde d_{l+1-\gamma}+\widetilde d_{l+1-\delta}}
    =
    \frac{g_{(l+r)}^{l+r+1-\gamma,l+r+1-\delta}(\widetilde t)}{\widetilde d_{l+r+1-\gamma}+\widetilde d_{l+r+1-\delta}}.
\]
One notes immediately that the denominators on the both sides coincide. The numerator on the left hand side is expressed by
\begin{equation}\label{eq: final expression in NN}
     g_{(l)}^{l+1-\gamma,l+1-\delta} = \sum_{i=1}^\gamma\sum_{j=1}^\delta \frac{\p t^{l+1-\gamma}}{\p y^{l+1-i}} \frac{\p t^{l+1-\delta}}{\p y^{l+1-j}} g^{l+1-i,l+1-j}_{(l)},
\end{equation}
where we follow Convention~\ref{convention on indices}.

The following lemma is the key to the stabilization in the NN sector.
Let $\widetilde y$ be the new coordinates s.t. $\widetilde y^i = y^{k+1-i}$. 
It follows immediately from Eq.~\eqref{eq: flat negative} that in the flat coordinates this is equiavalent to $\widetilde t^i = t^{k+1-i}$.

\begin{lemma}\label{prop: g stabilization}
    For any $r \ge 1$ and all $i,j$, s.t. $i+j \le l$
    \[
        g_{(l)}^{l+1-i,l+1-j}(\widetilde y) = g_{(l+r)}^{l+r+1-i,l+r+1-j}(\widetilde y).
    \]
\end{lemma}
\begin{proof}
Denote by $W_l$ the left hand side generating function of Eq~\eqref{eq: g via generating function} written in coordinates $\widetilde y^\bullet$. We are to show that the coefficient of $u^{2(i-1)}v^{2(j-1)}$ in $W_l$ equals the coefficient of $u^{2(i+r-1)}v^{2(j+r-1)}$ in $W_{l+r}$. This is equivalent to
\begin{equation}
    W_l = W_{l+r} \quad \text{modulo } u^{2a}v^{2b} \ \text{ s.t. } a+b \ge 2l.
\end{equation}

However $P_l = u^{2l} + u^{2(l-1)}\widetilde y^{k} + \dots + u^2 \widetilde y^{k-l} + \widetilde y^{k+1-l}$ and therefore
\begin{align*}
    & P_l = P_{l+r} \mid_{\widetilde y^a = 0, \ a > k}.
    \\
    &\qquad \Leftrightarrow \ P_l - P_{l+r} \equiv 0 \quad \text{modulo } u^{2l-1}.
\end{align*}
The lemma follows now from Eq.~\eqref{eq: g via generating function}.
\end{proof}
\begin{corollary}
    Under the conditions of the lemma above we have
    \[
        \eta_{(k,l)}^{l+1-i,l+1-j}(\widetilde y) = \eta_{(k+r,l+r)}^{l+r+1-i,l+r+1-j}(\widetilde y).
    \]
\end{corollary}
\begin{proof}
    This follows by differentiating the statement of the lemma above w.r.t. $\widetilde y^1$.
    The metric $\eta_{(k,l)}$ is obtained from $g_{(l)}$ by differentiating w.r.t. $y^k = \widetilde y^{1}$.
\end{proof}

The condition on the indices of the lemma above hold true for Eq.~\eqref{eq: final expression in NN} because of the bounds on $\gamma+\delta$ and $k$.

It follows from the corollary above that the matrices $\lbrace \p t^{l+1-\gamma} / \p y^{l+1-i} \rbrace_{\gamma,i=1}^{l-k}$ also stabilize. In particular, the flat coordinates $t^{k+1},
\dots,t^{l}$, expressed via $\widetilde y$ coincide for $B_{k,l}$ and $B_{k+r,l+r}$. This can be seen from Eq.~\eqref{eq: flat negative} as well.

Summing up we see that all three multiples of the sum in Eq.~\eqref{eq: final expression in NN} stabilize. This completes the proof.

\subsubsection{Sector PN}
Assume $\alpha,\mu,\nu$ are s.t. $1 \le \alpha \le k$ and $0 \le \mu,\nu \le l-k-1$. The potential $\F_{k,l}$ is at least cubic in its variables and the stabilization in the PN sector is equivalent to
\begin{align}
    &\left[ \frac{\p}{\p t^{k+1+\mu}} \frac{\p^2 \F_{k,l}}{\p t^{k+1-\alpha} \p t^{k+1+\nu}} \right](\widetilde t)
    =
    \left[ \frac{\p}{\p t^{r+k+1+\mu}} \frac{\p^2 \F_{k+r,l+r}}{\p t^{k+r+1-\alpha} \p t^{k+r+1+\nu}} \right](\widetilde t)
    \\
    &
    \left[ \frac{\p}{\p t^{k+1+\mu}} \frac{\p^2 \F_{k,l}}{\p t^{k+1-\alpha} \p t^{k+1+\nu}} \right](\widetilde t)
    =
    \left[ \frac{\p}{\p t^{r+k+1+\mu}} \frac{\p^2 \F_{k+r,l+r}}{\p t^{k+r+1-\alpha} \p t^{k+r+1+\nu}} \right](\widetilde t)
    \label{eq: PN second equation}
\end{align}
where $\widetilde t = \widetilde t(t)$ is the change of the variables given by $t^\epsilon = 4k \cdot \widetilde t^\epsilon$ if $1 \le \epsilon \le k$ and $t^l = 2 \cdot \widetilde t^{k+1-l}$, $t^{k+1} = 2 \cdot \widetilde t^0$,  $t^\nu = 4(l-k) \cdot \widetilde t^{k+1-\nu}$ for $\nu \ge k+1$.

The first equality above is equivalent to
\[
    \left[
    \frac{\p }{\p t^{k+1-\alpha}} \frac{g_{(l)}^{l-\mu,l-\nu}}{\widetilde d_{l-\mu} + \widetilde d_{l - \nu}}
    \right](\widetilde t)
    =
    \left[
    \frac{\p }{\p t^{k+r+1-\alpha}} \frac{g_{(l+r)}^{l+r-\mu,l+r-\nu}}{\widetilde d_{l+r-\mu} + \widetilde d_{l+r - \nu}}
    \right](\widetilde t)
\]

For the denominator we have
\[
 \widetilde d_{l-\mu} + \widetilde d_{l - \nu}
 = \frac{1+\mu+ \nu}{l-k}
 = \widetilde d_{l+r-\mu} + \widetilde d_{l+r-\nu}
\]
and we consider the numerator in details.

Recall Convention~\ref{convention on indices}.
Compute
\begin{align}\label{eq: final expression in PN}
    \frac{\p}{\p t^{k+1-\alpha}} g_{(l)}^{l-\mu,l-\nu}
    =
    \sum_{p=1}^{\nu+1}\sum_{q=1}^{\mu+1}
    \sum_{d=l-k+\alpha}^l
    \frac{\p y^{l+1-d}}{\p t^{k+1-\alpha}} \frac{\p t^{l-\nu}}{\p y^{l+1-p}} \frac{\p t^{l-\mu}}{\p y^{l+1-q}} \frac{\p g^{l+1-p,l+1-q}_{(l)}}{\p y^{l+1-d}}.
\end{align}
Due to the quasihomogeneity of $g$ all the summands above with $d \ge p+q-(l+1)$ are zero.
It was proved in Section~\ref{section: NN proof} that $\p t^{l-\nu} / \p t^{l+1-p}$ and $\p t^{l-\mu} / \p t^{l+1-q}$ stabilize. It was proved in \cite{BDbN} that $\p y^{l+1-d} / \p t^{k+1-\alpha}$ stabilize too. It remains to consider the fourth multiple above.

The expression of Eq.~\eqref{eq: final expression in PN} is zero by Proposition~\ref{prop: 3rd derivatives zeros} unless $l-\nu + l-\mu \le k+l+1$.
For this range we have $p+q \le l+1-k$.

\begin{lemma}
    Let $\widetilde y$ be as in Section~\ref{section: NN proof} and $p,q$ be s.t. $p+q < l+1$ then
    \[
      \frac{\p g_{(l)}^{l+1-p,l+1-q}}{\p y^{l+1-d}} (\widetilde y)
      =
      \frac{\p g_{(l+r)}^{l+r+1-p,l+r+1-q}}{\p y^{l+r+1-d}} (\widetilde y)
    \]
    as the functions of $\widetilde y$.
\end{lemma}
\begin{proof}
    This follows immediately from Eq.~\eqref{eq: g via generating function}.
\end{proof}

The partial derivatives in the lemma above are linear functions in $y^{k+1},\dots,y^l$. The change of the variables $y \to \widetilde y$ in this range is induced by $t \to \widetilde t$.
This lemma shows that all the four multiples in Eq~\eqref{eq: final expression in PN} stabilize in $\widetilde t$ coordinates.

The proof of Eq.~\eqref{eq: PN second equation} is skipped because it goes completely parallel.

\subsection{Proof of Theorem~\ref{theorem: consistent}}\label{section: proof 2}
Part (a) follows immediately from \cite{BDbN}[Theorem~6.3] and Corollary~\ref{corollary: Bk in Bkl}.

To prove part (b) we have to show that $\p_\gamma(\p_\alpha\p_\beta f) = \p_\alpha (\p_\beta \p_\gamma f)$ for any $\alpha,\beta,\gamma \in \ZZ$.

Being the potential of a Dubrovin--Frobenius manifold, $\F_{k,l}$ satisfies
\begin{equation}
    \sum_{\nu,\mu} \frac{\p^3 \F_{k,l}}{\p v^\sigma \p v^\gamma \p v^\nu} \eta^{\nu \mu} \frac{\p^3 \F_{k,l}}{\p v^\mu \p v^\alpha \p v^\beta}
    =
    \sum_{\nu,\mu} \frac{\p^3 \F_{k,l}}{\p v^\alpha \p v^\gamma \p v^\nu} \eta^{\nu \mu} \frac{\p^3 \F_{k,l}}{\p v^\mu \p v^\sigma \p v^\beta}
\end{equation}
for any fixed indices $\alpha,\beta,\sigma,\gamma$.

For given $\alpha,\beta,\gamma$ let $k$ and $l$ be s.t. the conditions of Theorem~\ref{theorem: consistent} hold true for any pair of indices from the three given.
Express the partial derivative $\p_\gamma (\p_\alpha \p_\beta f)$:
\begin{align*}
    \p_\gamma (\p_\alpha \p_\beta f) &= \sum_{\nu,\mu} \p_\gamma \p_1 \p_\nu f \cdot \eta^{\nu \mu} \cdot \frac{\p^3 \F_{k,l}}{\p v^\mu \p v^\alpha \p v^\beta} \mid_{v^\epsilon = \eta^{\epsilon\omega} \p_1\p_\omega f}
    \\
    & = \sum_{\nu,\mu} \p_1 \left( \frac{\p^2 \F_{k,l}}{\p v^\gamma \p v^\nu} \mid_{v^\epsilon = \eta^{\epsilon\omega} \p_1\p_\omega f} \right) \cdot \eta^{\nu \mu} \cdot \frac{\p^3 \F_{k,l}}{\p v^\mu \p v^\alpha \p v^\beta} \mid_{v^\epsilon = \eta^{\epsilon\omega} \p_1\p_\omega f}
    \\
    & = \sum_{\delta,\sigma} \p_1 \p_1\p_\delta f \cdot \eta^{\delta\sigma} \left( \sum_{\nu,\mu} \frac{\p^3 \F_{k,l}}{\p v^\sigma \p v^\gamma \p v^\nu}  \cdot \eta^{\nu \mu} \cdot \frac{\p^3 \F_{k,l}}{\p v^\mu \p v^\alpha \p v^\beta} \right) \mid_{v^\epsilon = \eta^{\epsilon\omega} \p_1\p_\omega f}.
\end{align*}
The expression obtained is symmetric in $\alpha,\beta,\gamma$ by WDVV equation above.

For part (c) we have to show that
\begin{equation*}
    \frac{\p^2 \F_{k,l}}{\p v^\alpha \p v^\beta} = \frac{\p^2 \F_{k,l}}{\p v^\alpha \p v^\beta} \mid_{v^\gamma = \eta^{\gamma\delta} \p_1 \p_\delta \F_{k,l}}.
\end{equation*}
The substitution on the RHS is
\[
 v^\gamma = \eta^{\gamma\delta} \frac{\p^2 \F_{k,l}}{\p v^1 \p v^\delta} = v^\gamma
\]
and the statement follows trivially.

To prove part (d) we make use of Propositions~\ref{proposition: positive to negative} and \ref{proposition: positive rescaling}. We show that Eq.~\eqref{eq: PDEs} with both negative indices are obtained from Eq.~\eqref{eq: PDEs} with both positive indices by the polynomial change of the initial condition data.

Consider any $\alpha,\beta > 0$ and let $k$, be s.t. $ k+1 \ge \alpha+\beta$. Then the positive flows can be written by
\begin{align*}
    & \p_\alpha \p_\beta f = \frac{\p^2 \F_{k,2k}}{\p t^{k+1-\alpha} \p t^{k+1-\beta}} \mid_{t^{k+1-\gamma} = \eta^{\gamma\delta} \p_1 \p_\delta f} 
    =
    \left( (t^l)^{2(\alpha+\beta-1)}\frac{\p^2 \F_{k,2k}}{\p t^{k+1-\alpha}\p t^{k+1-\beta}}  
    \right)
    \mid_{ t^{\gamma} = 4k \p_1 \p_{k+1-\gamma}f \cdot (\p_1 \p_{1-k} f)^{2 \gamma}}
    \\
    &=
    \left( \frac{\p^2 \F_{k,2k}}{\p t^{k+1+\alpha} \p t^{k+1+\beta}} (\widehat t) \right) 
    \mid_{ t^{\gamma} = 4k \p_1 \p_{k+1-\gamma}f \cdot (\p_1 \p_{1-k} f)^{2 \gamma}}
    =
    \p_{1-\alpha}\p_{1-\beta} f \mid_{ \p_1\p_\sigma f = \widehat t^\sigma(4k \p_1 \p_{k+1-\gamma}f \cdot  (\p_1 \p_{1-k} f)^{2 \gamma})}.
\end{align*}
where we have applied Proposition~\ref{proposition: positive rescaling} on the first line, Proposition~\ref{proposition: positive to negative} on the second line and also the polynomials $\widehat t^\sigma$, giving the change of the variables $\widehat t^\sigma = \widehat t^\sigma(t)$. Note that is the stabilization conditions hold true for $\alpha,\beta >0$, they also hold true for $1-\alpha,1-\beta \le 0$. Thus we got the relation between the positive flows $\p_\alpha \p_\beta f$ and the negative flows $\p_{1-\alpha} \p_{1-\beta} f$. This relation obtained by subtituting the initial data $\left\{ \p_1\p_\gamma f\right\}$.

\subsection{Examples}
Let $f_{\alpha,\beta} := \frac{\p^2 f}{\p t_\alpha \p t_\beta}$ as in Introduction.
The substitution 
\begin{align*}
    & f_{1,1}= \frac{f_{1,-4}}{2 f_{1,-5}}, \quad f_{1,2}= \frac{f_{1,-3}}{2 f_{1,-5}}-\frac{f_{1,-4}^2}{2 f_{1,-5}^2}, \quad f_{1,3}= \frac{f_{1,-4}^3}{2 f_{1,-5}^3}-\frac{f_{1,-3} f_{1,-4}}{f_{1,-5}^2}+\frac{f_{1,-2}}{2 f_{1,-5}},
    \\
    & f_{1,4}= -\frac{f_{1,-4}^4}{2 f_{1,-5}^4}+\frac{3 f_{1,-3} f_{1,-4}^2}{2 f_{1,-5}^3}-\frac{f_{1,-2} f_{1,-4}}{f_{1,-5}^2}+\frac{f_{1,-1}}{2 f_{1,-5}}-\frac{f_{1,-3}^2}{2 f_{1,-5}^2}, 
    \\
    &f_{1,5}= \frac{f_{1,-4}^5}{2 f_{1,-5}^5}-\frac{2 f_{1,-3} f_{1,-4}^3}{f_{1,-5}^4}+\frac{3 f_{1,-2} f_{1,-4}^2}{2 f_{1,-5}^3}+\frac{3 f_{1,-3}^2 f_{1,-4}}{2 f_{1,-5}^3}-\frac{f_{1,-1} f_{1,-4}}{f_{1,-5}^2}+\frac{f_{1,0}}{2 f_{1,-5}}-\frac{f_{1,-3} f_{1,-2}}{f_{1,-5}^2}
\end{align*}
transforms the flows
\begin{align*}
 & f_{0,0}=\frac{f_{1,-4}}{2 f_{1,-5}}, 
 \quad f_{0,-1}=\frac{f_{1,-3}}{2 f_{1,-5}}-\frac{f_{1,-4}^2}{2 f_{1,-5}^2},
 \quad f_{0,-2}=\frac{f_{1,-4}^3}{2 f_{1,-5}^3}-\frac{f_{1,-3} f_{1,-4}}{f_{1,-5}^2}+\frac{f_{1,-2}}{2 f_{1,-5}},
 \\
 & \qquad f_{0,-3}=-\frac{f_{1,-4}^4}{2 f_{1,-5}^4}+\frac{3 f_{1,-3} f_{1,-4}^2}{2 f_{1,-5}^3}-\frac{f_{1,-2} f_{1,-4}}{f_{1,-5}^2}+\frac{f_{1,-1}}{2 f_{1,-5}}-\frac{f_{1,-3}^2}{2 f_{1,-5}^2},
 \\
 & f_{0,-4}=\frac{f_{1,-4}^5}{2 f_{1,-5}^5}-\frac{2 f_{1,-3} f_{1,-4}^3}{f_{1,-5}^4}+\frac{3 f_{1,-2} f_{1,-4}^2}{2 f_{1,-5}^3}+\frac{3 f_{1,-3}^2 f_{1,-4}}{2 f_{1,-5}^3}-\frac{f_{1,-1} f_{1,-4}}{f_{1,-5}^2}+\frac{f_{1,0}}{2 f_{1,-5}}-\frac{f_{1,-3} f_{1,-2}}{f_{1,-5}^2},
 \end{align*}
 into the flows $f_{1,\alpha} = f_{1,\alpha}$ with $\alpha =1,\dots,5$ and the flows
 \begin{align*}
 &f_{-1,-1}=\frac{7 f_{1,-4}^3}{6 f_{1,-5}^3}-\frac{3 f_{1,-3} f_{1,-4}}{2 f_{1,-5}^2}+\frac{f_{1,-2}}{2 f_{1,-5}},
 \\
 & f_{-1,-2}=-\frac{2 f_{1,-4}^4}{f_{1,-5}^4}+\frac{4 f_{1,-3} f_{1,-4}^2}{f_{1,-5}^3}-\frac{3 f_{1,-2} f_{1,-4}}{2 f_{1,-5}^2}+\frac{f_{1,-1}}{2 f_{1,-5}}-\frac{f_{1,-3}^2}{f_{1,-5}^2},
 \\
 & f_{-1,-3}=\frac{3 f_{1,-4}^5}{f_{1,-5}^5}-\frac{17 f_{1,-3} f_{1,-4}^3}{2 f_{1,-5}^4}+\frac{4 f_{1,-2} f_{1,-4}^2}{f_{1,-5}^3}+\frac{9 f_{1,-3}^2 f_{1,-4}}{2 f_{1,-5}^3}-\frac{3 f_{1,-1} f_{1,-4}}{2 f_{1,-5}^2}+\frac{f_{1,0}}{2 f_{1,-5}}-\frac{2 f_{1,-3} f_{1,-2}}{f_{1,-5}^2}
\end{align*}
into those given in Eq.~\eqref{eq: first positive flows}.

\end{document}